\documentclass[a4paper,11pt]{article}
\usepackage{graphicx,amssymb,amsmath,amsfonts}
\usepackage{amsthm}
\usepackage{url}
\usepackage{hyperref}
\usepackage{color}
\usepackage{fullpage}
\usepackage[T1]{fontenc}
\usepackage{lineno}
\usepackage{url}
\usepackage{framed}
\usepackage{boxedminipage}
\usepackage{enumerate}
\usepackage{textcomp}
\usepackage{cite}
\usepackage[lined,boxed,commentsnumbered,ruled,vlined,noend,linesnumbered,boxed]
{algorithm2e}

\newtheoremstyle{mytheorem}{3pt}{3pt}{\slshape}{}{\bfseries}{}{.5em}{}
\theoremstyle{mytheorem}
\newtheorem{lemma}{Lemma}
\newtheorem{theorem}{Theorem}

\newtheorem{observation}{Observation}

\newtheorem{remark}{Remark}
\theoremstyle{definition}

\def\ann{$\mathcal{A}$}

\def\denseitems{
    \itemsep1pt plus1pt minus1pt
    \parsep0pt plus0pt
    \parskip0pt\topsep0pt}
    
\newbox\ProofSym \setbox\ProofSym=\hbox{%
  \unitlength=0.18ex%
  \begin{picture}(10,10) \put(0,0){\framebox(9,9){}}
    \put(0,3){\framebox(6,6){}}
  \end{picture}}

\title{Maximum-Width Rainbow-Bisecting Empty Annulus\thanks{%
A preliminary version is accepted in EuroCG 2021 and the expanded version is accepted in the journal Computational Geometry: Theory and Applications.}}

\author{Sang Won Bae$^1$\thanks{supported by Basic Science Research Program through the National Research Foundation of Korea (NRF) funded by the Korea government(MSIT)(No.
RS-2023-00251168).}  Sandip Banerjee$^2$\thanks{acknowledges the support by Polish National Science Centre (NCN) Grant 2020/39/B/ST6/01641.}
Arpita Baral$^3$  Priya Ranjan Sinha Mahapatra$^4$\thanks{acknowledges the support received from
MATRICS (Ref. No. MTR/2019/000792) grant of Science and Engineering Research Board (SERB), Government of India.}  Sang Duk Yoon$^5$}
\date{%
    $^1$Division of Computer Science and Engineering, Kyonggi University, Korea.\\%
    $^2$ Institute of Informatics, University of Wroc\l aw, Poland.\\%
    $^3$ Department of Computer Science and Engineering,
    NSHM Knowledge Campus (Group of Institutions), Durgapur, India.\\
    $^4$ Department of Computer Science and Engineering, University of Kalyani, India.\\
    $^5$ Department of Service and Design Engineering, Sungshin Women\textquotesingle s University, Korea.\\
    \today
}

\begin{document}
\maketitle              

\begin{abstract}
Given a set of $n$ colored points with $k$ colors in the plane,
we study the problem of computing a maximum-width rainbow-bisecting empty annulus (of objects specifically axis-parallel square, axis-parallel rectangle and circle) problem. We call a region \emph{rainbow} if it contains
at least one point of each color. 
The maximum-width rainbow-bisecting empty
annulus problem asks to find an annulus $A$ 
of a particular shape with maximum possible width such that $A$ does not contain any input points and it bisects the input point set into two parts, each of which is a \emph{rainbow}.
We compute a maximum-width rainbow-bisecting empty axis-parallel 
square, axis-parallel rectangular and circular annulus in $O(n^3)$ time using $O(n)$ space, 
in $O(k^2n^2\log n)$ time using $O(n\log n)$ space
and in $O(n^3)$ time using $O(n^2)$ space respectively. 
\end{abstract}

\section{Introduction}
In the context of facility location most of the existing literature deals with the placement of the
facilities (e.g. pipelines) among a set of customers (represented by points in $\mathbb{R}^2$)
but there are scenarios where the facilities are hazardous (e.g pipelines transporting toxic materials). In these scenarios the objective is maximizing the minimal distance between the hazardous facility
and the given customers. 
Considering this situation, some problems have been studied in literature that computes a widest $L$-shaped empty corridor~\cite{c-welsc-96}, a widest 1-corner corridor~\cite{dls-fwe1cc-06}, a largest empty annulus~\cite{dhmrs-leap-03}.
For empty annulus we additionally impose another constraint that each of the two non-empty regions
corresponds to two \emph{self-sustained} smart cities, meaning that each of the smart cities is comprised of essential facilities of each type such as schools, hospitals, etc. This situation calls for the computation of two rainbow regions separated by an empty region where empty region contains hazardous facility and each color represents
each essential facilities in each of the rainbow region.
Given a set of $n$ points in $\mathbb{R}^2$, each point colored with one of the $k$ ($k\leq n$) colors, a \emph{rainbow} (or color spanning) region in $\mathbb{R}^2$ contains at least one point of each color. Abellenas et al. first proposed algorithms for computing a smallest color-spanning axis-parallel rectangle~\cite{ahiklmps-scso-2001}, narrowest strip~\cite{ahiklmps-scso-2001}, circle~\cite{ahiklmps-tfcvdarp-2001}  in $O(n(n-k)\log^2 k)$, $O(n^2\alpha(k)\log k)$ and $O(kn\log n)$ respectively.
 Later the time complexities to compute a 
smallest color-spanning rectangle and a narrowest color-spanning strip were improved
to $O(n(n-k)\log k)$ and $O(n^2\log n)$ by Das et al.~\cite{dgn-scsor-09}.
Kanteimouri et al.~\cite{kmak-ctscsaps-2013} computed a smallest axis-parallel color-spanning square in $O(n \log^2 n)$ time. Hasheminejad et.al~\cite{hkm-ctscset-2015} proposed an $O(n \log n)$ time algorithm to compute a smallest
color-spanning equilateral triangle.
The shortest color-spanning interval, smallest color-spanning $t$ squares and 
 smallest color-spanning $t$ circles have been studied by Banerjee et al.~\cite{bmn-csoaahr-2016} and they have given some
hardness and tractability results from parameterized complexity point of view. 
Acharyya et al.~\cite{anr-mwcsa-18} computed a minimum-width color-spanning annulus for circle, equilateral triangle, axis-parallel square and rectangle in $O(n^3\log n)$, $O(n^2k)$, $O(n^3+n^2k\log k)$ and $O(n^4)$ time respectively. D\`{i}az-B\'{a}\~{n}ez
et al.~\cite{dhmrs-leap-03} first studied the largest empty circular annulus problem and
proposed an $O(n^3\log n)$ time algorithm.
 They later improved the time complexity to $O(n^3)$ time~\cite{bls-laop-06}. 
 Bae et al.~\cite{bbm-mwesra-2021} proposed algorithms for computing a maximum-width empty axis-parallel 
 square and rectangular annulus in $O(n^3)$ and $O(n^2 \log n)$ time respectively. Recently Erkin et al.~\cite{ABCCKMS18} studied a variant of a covering problem where the objective is to cover a set of points by a conflict free set of objects, where an object is said to be \emph{conflict free} if it covers at most one point from each color class.\\
 
In this paper,  we mainly address three problems - (I) computing a maximum-width rainbow-bisecting empty 
 axis-parallel square annulus, (II) computing a maximum-width rainbow-bisecting empty 
 axis-parallel rectangular annulus and (III) computing a maximum-width rainbow-bisecting empty circular annulus.
 For the last problem we also consider an additional case where the center of the empty circular annulus will lie on a given query line in the plane.
Problem I is solved in $O(n^3)$ time using $O(n)$ space.
Note that our Problem I generalizes the problem of computing 
a largest \emph{rainbow-bisecting empty L-shaped corridor}.
Hence, we also consider the corridor problem as a sub-case of Problem I
and present an $O(n \log n)$-time algorithm.
This siginificantly improves the time complexity of the $O(n^2\log n)$-time algorithm for
the widest empty axis-parallel L-shaped corridor problem~\cite{bbm-mwesra-2021}.
For problems II and III, we propose algorithms that run in $O(k^2n^2\log n)$ time using $O(n\log n)$ space
and in $O(n^3)$ time using $O(n^2)$ space respectively.

\section{ Definition and Terminologies}\label{sec:problemdef}
We are given a set $P$ containing $n$ points in $\mathbb R^{2}$ where each 
point is colored with one of the given colors $\{1, \ldots, k\}$ and 
$k \leq n/2$. 
For each color $i\in [k]$,
we assume that there are at least two points in $P$ of color $i$. 
For any point $p \in P$,
we denote its $x$-, $y$-coordinates, and color by $x(p)$, $y(p)$ and $\alpha(p)$ respectively.
For any two points $p$ and $q$ in $P$, $x$- and $y$-gap
between $p$ and $q$ is defined as $|x(p)-x(q)|$ and $|y(p)-y(q)|$, respectively.

Let $S$ be an axis-parallel square.
The four sides of $S$ are called the top, bottom, left, and right sides, respectively,
according to their relative position in $S$.
The center of $S$ is defined as the intersection between two diagonals 
of $S$, and the radius of $S$ is defined as half of its side length.
The \emph{offset} of $S$
by a real number $\delta \geq 0$ is defined as a smaller axis-parallel square
obtained by sliding each side of $S$ inwards by $\delta$.

An \emph{axis-parallel square annulus} (Figure~\ref{fig:annulus}(a)) \ann~ is the region between an axis-parallel square $S_\mathrm{out}$ and its offset $S_\mathrm{in}$ 
by a real number $\delta \geq 0$.
We call $S_\mathrm{out}$, $S_\mathrm{in}$, and $\delta$ the \emph{outer} square, \emph{inner} square,
and the \emph{width} of the annulus, respectively.
Note that $S_\mathrm{out}$ and $S_\mathrm{in}$ are concentric, so
the center of $S_\mathrm{out}$ and $S_\mathrm{in}$ can be treated as
the center of the annulus.
We allow the outer and inner squares of an
axis-parallel square annulus \ann~ to
have one or more sides \emph{at infinity} ($\pm\infty$)
which means the associated coordinate value of that side is infinite. 

\begin{figure}[h]
	\centering
	\includegraphics[width=.55\textwidth]{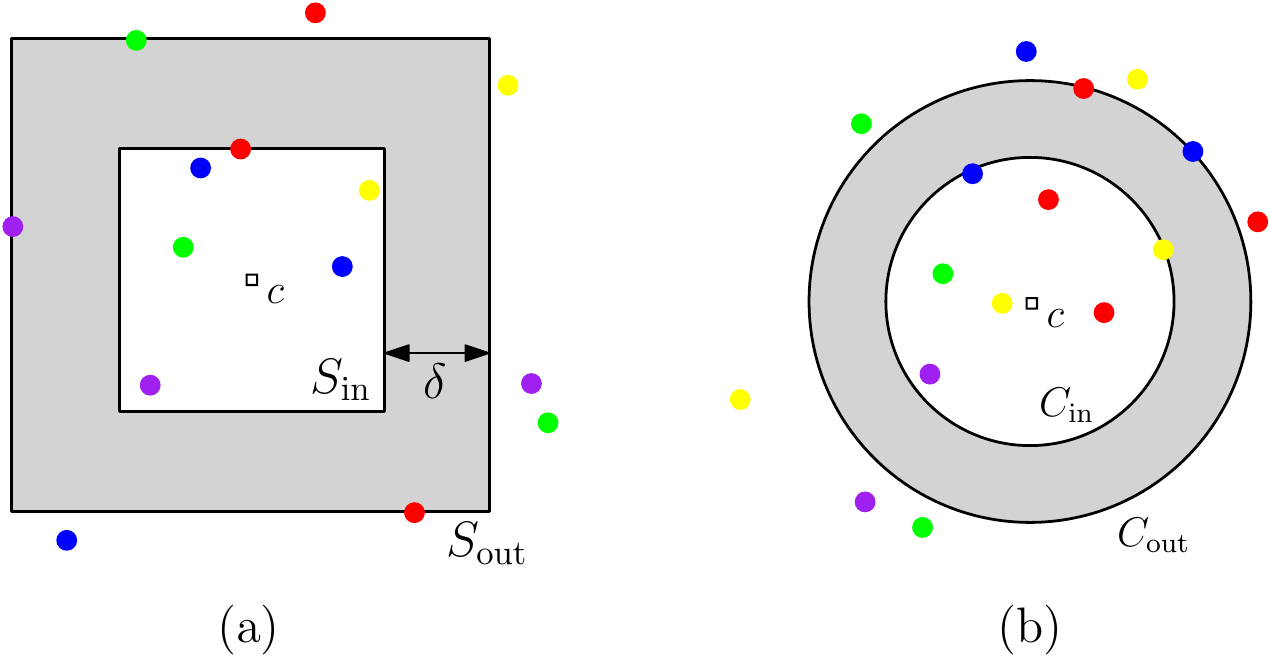}
	\caption{(a) A RBSA of width $\delta$
		with outer and inner squares $S_\mathrm{out}$ and $S_\mathrm{in}$.
		(b) A RBCA with outer and inner circles $C_\mathrm{out}$ and $C_\mathrm{in}$.
		For both annuli, $c$ denotes the center.}
	\label{fig:annulus}
\end{figure}

For the rectangular annulus, we consider the definition of Mukherjee
et al.~\cite{jmkd-mwra-2012}.
An \emph{axis-parallel rectangular annulus} is the region obtained by subtracting 
the interior of an axis-parallel rectangle $R_{in}$ from another axis-parallel rectangle 
$R_\mathrm{out}$ such that $R_\mathrm{in} \subseteq R_\mathrm{out}$.
We call $R_\mathrm{out}$ and $R_\mathrm{in}$ the \emph{outer rectangle} and \emph{inner rectangle} of
the annulus, respectively.
Consider a rectangular annulus \ann~defined by its outer and inner rectangles, $R_\mathrm{out}$ and $R_\mathrm{in}$ respectively.
By our definition, note that $R_\mathrm{out}$ and $R_\mathrm{in}$ defining annulus \ann~do not have to be concentric, 
so that \ann~may not be a symmetric shape.
The \emph{top-width} of \ann~is the vertical distance between the top sides of $R_\mathrm{out}$ and $R_\mathrm{in}$, 
and the \emph{bottom-width} of \ann~is the vertical distance between their bottom sides.
Analogously, the \emph{left-width} and \emph{right-width} of \ann~are defined to be 
the horizontal distance between the left sides of $R_\mathrm{out}$ and $R_\mathrm{in}$ and the right sides of $R_\mathrm{out}$ and $R_\mathrm{in}$, respectively.
Then, the \emph{width} of \ann~ is defined to be the minimum of the four values: 
the top-width, bottom-width, left-width, and right-width of \ann.
See \figurename~\ref{fig:rectann} for an illustration.
\begin{figure}[t]
	\centering
	\includegraphics[width=.35\textwidth]{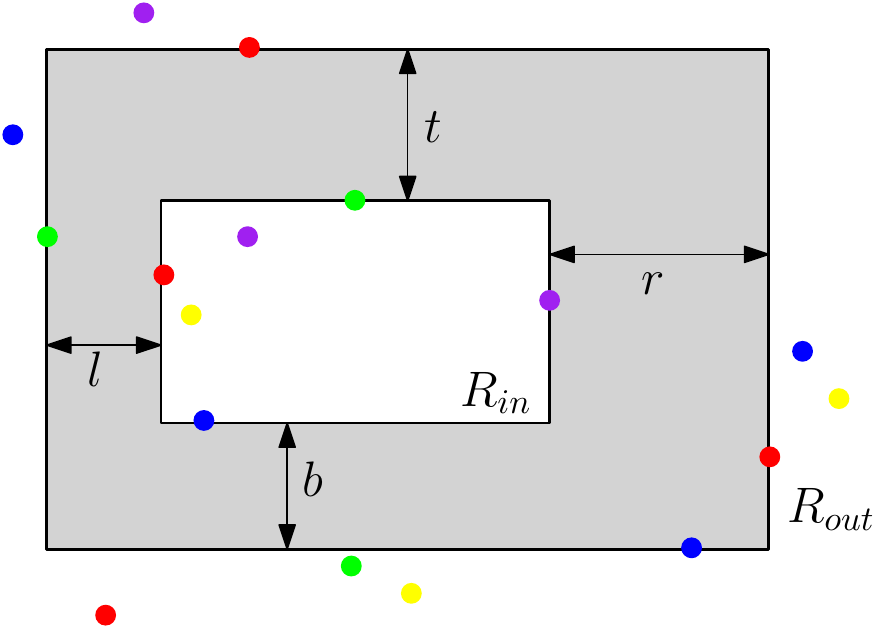}
	\caption{ A RBRA with outer and inner rectangles $R_\mathrm{out}$ and $R_\mathrm{in}$
		whose top-, bottom-, left-, right-widths are $t$, $b$, $l$, and $r$,
		respectively.}
	\label{fig:rectann}
\end{figure}

A \emph{circular annulus} \ann~ is the region between two concentric circles $C_\mathrm{out}$ (outer circle) and 
$C_\mathrm{in}$ (inner circle) where $C_\mathrm{in} \subseteq C_\mathrm{out}$.
The \emph{width} of a circular annulus \ann~ is defined to be the difference of the radii
of its outer and inner circles (see Figure~\ref{fig:annulus}(b)).

An  annulus \ann~is said to be
\emph{rainbow-bisecting empty} if it does not contain any point of $P$ 
in its interior and divides $P$ into two
non-empty subsets such that each subset is a rainbow.
One subset of $P$ lies within or on the boundary of the
inner square/rectangle/circle of \ann~and the other subset of $P$
lies outside or on the boundary of the outer square/rectangle/circle of \ann. 
We now refer to an axis-parallel square (resp. rectangle) as square (resp. rectangle)
for simplicity.
Any \emph{rainbow-bisecting empty square annulus},
\emph{rainbow-bisecting empty rectangular annulus}
and \emph{rainbow-bisecting empty circular annulus} are denoted by 
RBSA, RBRA and RBCA, respectively.

\section{Maximum-Width Rainbow-Bisecting Empty Square Annulus} \label{sec:sq}
In this section, we compute a maximum-width RBSA from a given point set $P$ on $\mathbb{R}^2$. 

The following lemma is analogous to Observation 1 in Bae et al.~\cite{bbm-mwesra-2021} but is introduced for the completeness of the paper.
\begin{lemma} \label{obs:sq_conf}
	There is a maximum-width RBSA such that the inner and the outer squares 
	each contains a point of $P$ on their boundaries
	and one of the following three conditions is satisfied.
	\begin{itemize}
		\item C$_{1}$: The sides of the outer square containing no point of $P$ 
		are at infinity (Figure~\ref{fig-cases-sq}(a)).
		\item C$_{2}$: Two adjacent sides of the outer square contains a point of $P$, 
		and the remaining sides at infinity (Figure~\ref{fig-cases-sq}(b)).
		\item C$_{3}$: Two opposite sides of the outer square contains a point of $P$
		(Figure~\ref{fig-cases-sq}(c)).
	\end{itemize}
\end{lemma}

\begin{figure}[h]
	\centering
	\includegraphics[width=.65\textwidth]{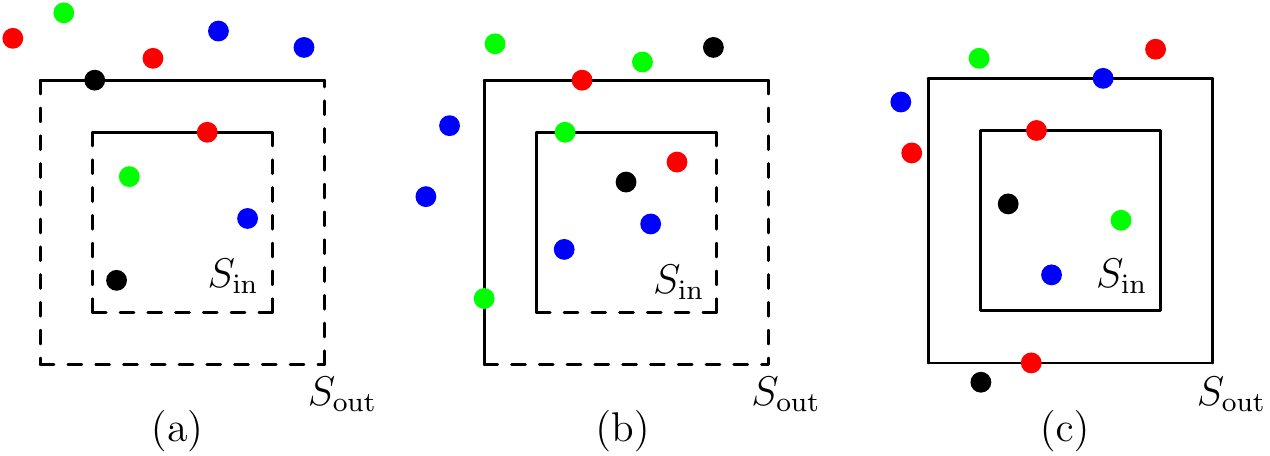}
	\caption{
		Three possible configurations of a maximum-width RBSA. 
		The dashed sides represent the sides at infinity.
		(a) C$_{1}$ Configuration (b) C$_{2}$ Configuration (c) C$_{3}$ Configuration}
	\label{fig-cases-sq}
\end{figure}

\begin{proof}
We start by proving the fact that the
	outer and inner squares of a maximum-width RBSA
	each contains a point of $P$ on their boundaries.
	For a contradiction,
	we assume that there is no point on the boundaries of 
	the outer square $S_\mathrm{out}$ and the
	inner square $S_\mathrm{in}$ of a maximum-width RBSA, \ann.
	Then, we can shrink $S_\mathrm{in}$ until the boundary 
	of $S_\mathrm{in}$ hits a point in $P$,
	while keeping the center of \ann~fixed.
	In this process no point goes outside from the inside of $S_\mathrm{in}$, so
	\ann~remains empty.
	But the width of \ann~increases which contradicts our assumption.
	Similarly, we can increase the width of  \ann~by extending 
	$S_\mathrm{out}$ while keeping the center of \ann~fixed, a contradiction.
	So there must be at least one point of $P$ lying on the boundary of $S_\mathrm{in}$
	and on the boundary of $S_\mathrm{out}$. Without loss of generality, we assume that the top side of $S_\mathrm{out}$
	contains of a point of $P$. 
	
	Consider the case that there is no point of $P$ on the remaining sides of
	$S_\mathrm{out}$. We can enlarge both $S_\mathrm{in}$ and $S_\mathrm{out}$
	while maintaining the width of \ann~to be fixed. See Figure~\ref{fig-subcases}(a).
	The point of $P$ on the boundary of $S_\mathrm{in}$ does not affect
	the enlarging process, and the process will be finished when $S_\mathrm{out}$ encounters a point of $P$. If there is no such a point, \ann~belong to C$_1$ configuration.
	
	Assume that the enlarging process is finished with a point $q$ of $P$.
	If $q$ is on the bottom side of $S_\mathrm{out}$, then \ann~belongs 
	to C$_3$ configuration. If $q$ is on the left side of $S_\mathrm{out}$, 
	we can enlarge both $S_\mathrm{in}$ and $S_\mathrm{out}$ further 
	as in Figure~\ref{fig-subcases}(b). If this process is finished with 
	another point of $P$, then \ann~belongs to C$_3$ configuration.
	If there is no such a point, then \ann~belongs to C$_2$ configuration.
	The case when $q$ is on the right side of $S_\mathrm{out}$ 
	can be handled similarly.
	
	\begin{figure}[h]
		\centering
		\includegraphics[width=.55\textwidth]{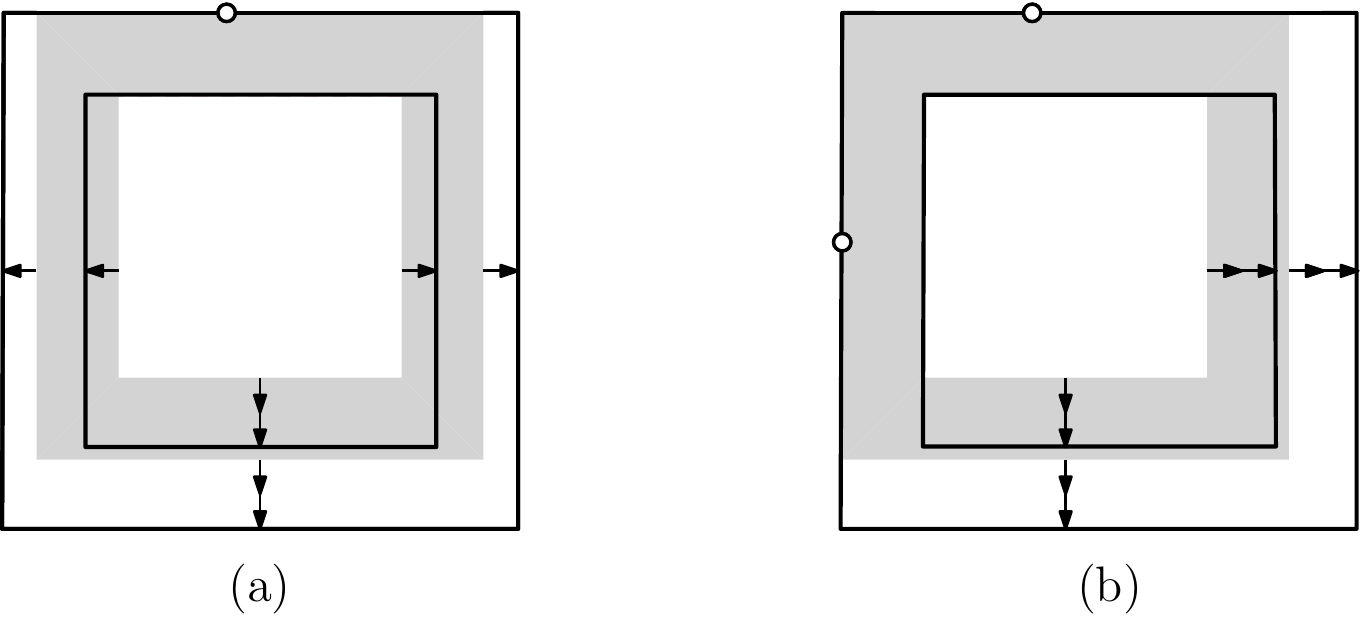}
		\caption{
			Examples of enlarging the inner and the outer squares of a RBSA (gray filled) 
			to get another RBSA (black) with the same width.
		}
		\label{fig-subcases}
	\end{figure}

	Next, we consider the case such that there is another point $q$ of $P$ 
	on a side other than the top side of $S_\mathrm{out}$.
	If $q$ is on the bottom side of $S_\mathrm{out}$, \ann~belongs 
	to C$_3$ configuration. If $q$ is on the left or right side of $S_\mathrm{out}$,
	then we can enlarge $S_\mathrm{in}$ and $S_\mathrm{out}$,
	and conclude that \ann~belongs to C$_2$ or C$_3$ configuration.
\end{proof} 

Thus to find a maximum-width RBSA from a given point set $P$, 
we solve the above three configurations and 
propose different algorithms for handling each of them.
As a part of preprocessing, we sort the points in $P$
using their $x$- and $y$-coordinates, respectively.

Note that any RBSA belonging to configuration C$_{1}$ maps to an empty strip such that  
regions on both side of the strip are rainbow.
We call such RBSA a \emph{rainbow-bisecting empty strip}, RBES.
 
\begin{lemma}\label{lem:cs_emp_sq}
A maximum-width RBES of a sorted set of points
can be computed in $O(n)$ time using $O(n)$ space.
\end{lemma}

Now we discuss the solution technique for C$_{2}$ configuration.
Observe solving this configuration corresponds to the problem of finding a 
widest \emph{rainbow-bisecting} empty axis-parallel L-shaped corridor (RBLC). 
 A RBLC is an axis-parallel empty L-shaped corridor
which partitions the input point set $P$ into two non-empty rainbow subsets.
Note that one can solve the problem by constructing $n^2$ grid points for given $n$ points and have the following trivial result.
\begin{remark}\label{lem:emp_sq_type2}
 Given a set of $n$ points in the plane, each colored with one of the given $k$ colors, a maximum-width RBLC can be computed in $O(n^3)$ time using $O(n)$ space.
 \end{remark}

 \begin{figure}[h]
	\centering
	\includegraphics[width=.8\textwidth]{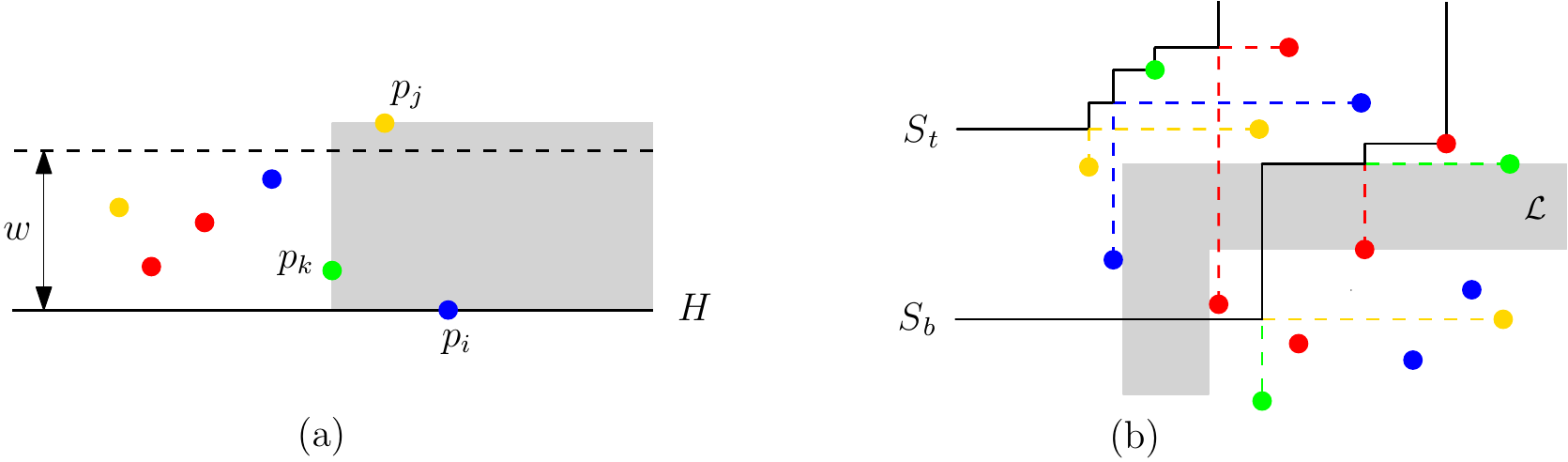}
	\caption{(a) The horizontal part of $\mathcal{L}$ should be contained in 
		the gray region bounded by $p_i$, $p_j$, and $p_k$.
		(b) Two corners of $\mathcal{L}$ should be 
			contained in the region bounded above S$_{t}$ and below S$_{b}$. }
	\label{fig:sqannulus1}
\end{figure}

We design a non-trivial algorithm for computing a maximum-width RBLC. 
Without loss of generality, we only consider the case where 
the RBLC is pointing down and right, and 
its width is determined from the $y$-gap of two points in $P$ 
(Figure~\ref{fig-cases-sq}(b)).
Observe that we can handle all other cases analogously.
Let $P = \{p_1, p_2, \ldots, p_n\}$ be the given set of points,
sorted in the ascending order of their $y$-coordinates,
that is, $y(p_1) \leq y(p_2) \leq \cdots \leq y(p_n)$.
We imagine sweeping a horizontal line $H$ upwards over the plane. 
It starts from the position $y = y(p_1)$ and gradually moves upwards until 
all points in $P$ are visited. 
When the line $H$ sweeps through a point $p_i \in P$, 
we decide the existence of a RBLC $\mathcal{L}$
such that the horizontal part of $\mathcal{L}$ is right above $H$, and
its width is larger than the width of the widest RBLC found so far.
To find such RBLC $\mathcal{L}$, we use the following queries.
Let $Q_a$ be the set of points above $y(p_i)$ and $Q_b$  be the set of remaining points 
$P \setminus Q_a$.

\begin{enumerate}[(i)] \denseitems
 \item \emph{Boundary points query}:  
 	Let $w$ be the width of the widest RBLC found so far.
 	First, we find the lowest point $p_j$ in $Q_a$ such that
 	the $y$-gap between $p_i$ and $p_j$ is larger than $w$.
 	Then $y$-gap between $p_i$ and $p_j$ becomes the width of $\mathcal{L}$.
 	Among the points whose $y$-coordinates are in between $y(p_i)$ and $y(p_j)$, 
 	we find the point $p_k$ with the largest $x$-coordinate.
 	Observe that the horizontal part of $\mathcal{L}$ should be contained in
 	the region bounded by $p_i$, $p_j$, and $p_k$ from
 	below, above, and left, respectively (Figure~\ref{fig:sqannulus1}(a)).

\item \emph{Rainbow range query}: 
Let S$_{b}$ and S$_{t}$ are two staircase structures such that 
any empty axis-parallel L-shaped corridor whose two corners are
lying in the region above S$_{b}$ and below S$_{t}$ is a RBLC
	(Figure~\ref{fig:sqannulus1}(b)).
	After the boundary points query, we compute the intersection 
	$\mathcal{I}_t$ (resp. $\mathcal{I}_b$)
	between the horizontal line $y = y(p_j)$ and S$_t$ (resp. $y = y(p_i)$ and S$_b$). 
	The left (resp. right) side of the vertical part 
	of $\mathcal{L}$ should be on the right (resp. left) of the intersection.

\item \emph{Maximum $x$-gap query}: 
	After the rainbow range query, we find the maximum $x$-gap
	between two consecutive points of $Q_b$ in range 
	$[\max (x(\mathcal{I}_t), x(p_k)),x(\mathcal{I}_b)]$,
	when the points of $Q_b$ are sorted with respect to the $x$-coordinates.
	If the maximum $x$-gap is larger than $y$-gap between $p_i$ and $p_j$, 
	then $\mathcal{L}$ exists.

\end{enumerate}

To answer the above queries we use the following data structures :
\begin{itemize}
 \item 
		The point $p_j$ can be found in $O(\log n)$ time without any additional 
		data structure.
		Let $\mathcal{T}$ be a one dimensional range tree built on 
		the $y$-coordinate values of the points in $P$. 
		Also each node $v\in \mathcal{T}$ stores the point with
		the maximum $x$-coordinate value among the points
		in the canonical subset of $v$.
		The tree $\mathcal{T}$ can be constructed in $O(n)$ time using $O(n)$ space~\cite{BCKO08}, and
		finding the points with the maximum $x$-coordinate values can be done 
		in $O(n)$ time in total. Using this data structure, we can find the point $p_k$
		in $O(\log n)$ time.

 \item 
		The corners of S$_{b}$ and S$_{t}$ are maintained in 
		two balanced binary search trees $\mathcal{T}_{b}$ and $\mathcal{T}_{t}$ 
		built on their $y$-coordinate values, respectively.
		We compute S$_{b}$ as follows, and S$_{t}$ can be computed similarly.
		Let $c$ be a point at $(-\infty, -\infty)$, and consider the region $R_c$ which is the intersection 
		between two half spaces $x\geq x(c)$ and $y\leq y(c)$.
		Now we move $c$ upward until $R_c$ becomes a rainbow set.
		Next, we move $c$ to the right while $R_c$ is a rainbow set.
		When $c$ is stopped, there should be a point $p$ on the left side of $R_c$, 
		and we move $c$ upward again until $R_c$ meets a point $q$ such that
		$\alpha(p)=\alpha(q)$. 
		Then, we move $c$ to the right again, while $R_c$ is a rainbow set.
		By repeating this process, we can get S$_{b}$ from the trace of $c$.
		\\
		\\
		As $P$ is sorted with respect to the $y$-coordinates, 
		the next point to be contained in $R_c$ (as $c$ moves upward)
		can be found in $O(1)$ time.
		When adding a new point to $R_c$,
		we spend $O(\log n)$ time to keep the points in $R_c$ sorted 
		with respect to the $x$-coordinates.
		Then we can find the next point to be removed from $R_c$ 
		(as $c$ moves to the right) in $O(\log n)$ time.
		In total, we can compute $\mathcal{T}_{b}$ that stores the corners of
		S$_{b}$ in $O(n\log n)$ time using $O(n)$ space.
		From $\mathcal{T}_{b}$ and $\mathcal{T}_{t}$, 
		rainbow range query can be answered in $O(\log n)$ time.
		
 \item Let $\mathcal{T}_{i}$ be a one dimensional range tree 
 build on the $x$-coordinate values of the points  in $Q_b$. 
Also each node $v\in \mathcal{T}_{i}$ stores the maximum 
 $x$-gap of two consecutive points in the canonical subset of $v$. 
 The structure $\mathcal{T}_{i}$ can be computed from $\mathcal{T}_{i-1}$
 by adding new point $p_i$ in $O(\log n)$ time, and the size of $\mathcal{T}_{i}$
 is bounded by $O(n)$~\cite{BCKO08}.
The maximum $x$-gap query for a given range can be answered in $O(\log n)$ time
 by comparing $O(\log n)$ values including every gap between two consecutive canonical subsets.
 
\end{itemize}


With the above data structures in hand, we have the following result.

\begin{lemma}\label{lem:cssq2}
The existence of a RBLC $\mathcal{L}$ can be determined 
(during the sweeping process) in $O(\log n)$ time
such that the horizontal part of $\mathcal{L}$ is right above the sweeping line and it's width is larger than the width of the widest RBLC found so far.
\end{lemma}
Our algorithm uses the result of  Lemma~\ref{lem:cssq2} 
to obtain a maximum-width RBLC from $P$.

\begin{theorem} \label{thm:cssq2}
 Given $n$ points in $\mathbb{R}^2$ 
 where each one is colored with one of the $k$ colors, a maximum-width 
 RBLC can be found in $O(n\log n)$ time using $O(n)$ space.
\end{theorem}
\begin{proof}
		The sweeping line $H$ starts from the position $y = y(p_1)$ and gradually moves upwards.
		When the sweeping line $H$ passes through a point $p_i$, we determine 
		the existence of a RBLC $\mathcal{L}$ such that the horizontal part of $\mathcal{L}$ is right above $H$ and its width is larger than the width of the widest RBLC found so far.
		We repeat this until there is no such RBLC, and move $H$ to $y = y(p_{i+1})$.
		The width of $\mathcal{L}$ is determined by the $y$-gap
		between two points $p_i$ and $p_j$ where $1\leq i<j\leq n$, and both indices $i$ and $j$
		do not decrease during the sweeping process.
		Therefore, the number of RBLCs we find during the sweeping process is $O(n)$,
		and we can find a maximum-width RBLC in $O(n\log n)$ time using $O(n)$ space.
\end{proof}

Here we mention that our algorithm improves the time complexity for computing a
\emph{widest empty axis parallel L corridor} for a given set of $n$ points in plane 
from $O(n^2\log n)$ (Theorem 1 of~\cite{bbm-mwesra-2021}) to $O(n\log n)$.\\

C$_{3}$ configuration: Any RBSA here is defined by two input points on 
both two opposite sides of its outer square. Assume that each of the top and bottom sides
of the outer square of a RBSA contain a point of $P$. We design an algorithm to compute a 
maximum-width RBSA with the the top and bottom sides of its outer square containing points $p_i$ and $p_j$,
respectively, from all pairs of indices $(i, j)$
with $1\leq i < j-1 < n$.
We start this case with the following observation.
\begin{observation}
For fixed points $p_{i}$ and $p_{j}$ which
define the bottom and top sides of the outer square $S$ respectively,
the potential locations of the center of $S$
lie on the line $\ell$ such that $y(\ell) = (y(p_{i})+y(p_{j}))/2$ (Ref. Figure~\ref{fig:cs_annulus}).
\end{observation}

Note that the outer square of a potential candidate annulus for fixed points $p_{i}$
and $p_{j}$ lies inside a rectangle $R$ (Ref. Figure~\ref{fig:cs_annulus}) such that
$x$(left)$=\max\{ x(p_{i}), x(p_{j})\} - 2r$ and 
$x$(right)$=\min\{x(p_{i}), x(p_{j})\} + 2r$, where $r = (y(p_{j})-y(p_{i}))/2$.

\begin{figure}[t]
 \centering
\includegraphics[width=.4\textwidth]{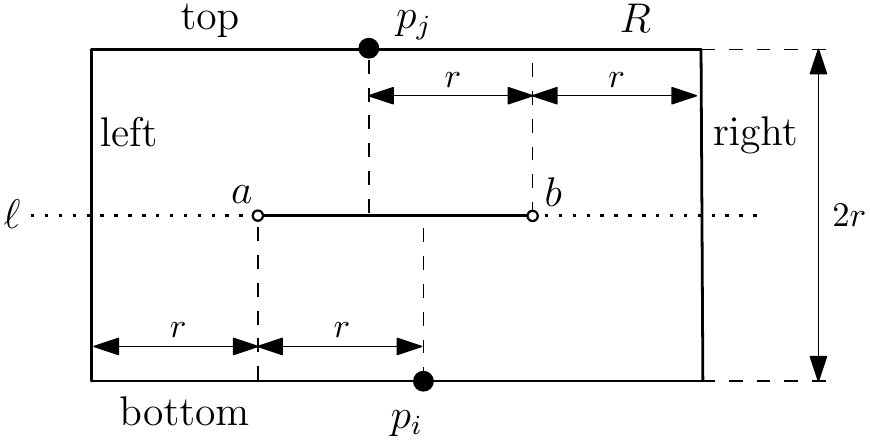}
\caption{
Centers of candidate outer squares in $R$
lie in the horizontal segment $[a,b]$ such that $a=(max(x(p_{i}), x(p_{j})) - r, y(\ell))$ 
and $b=min(x(p_{i}), x(p_{j})) + r, y(\ell))$, where $y(\ell) = (y(p_{i})+y(p_{j}))/2$ and
$r = (y(p_{j})-y(p_{i}))/2$.}
\label{fig:cs_annulus}
\end{figure}
We maintain color counter vectors to maintain the number of points of each color present outside and inside the 
horizontal strip defined by points $p_{i}$ and $p_{j}$. Since the target annulus is a RBSA, therefore rectangle 
$R$ must be a rainbow rectangle. We only consider such rectangles.

For a fixed outer square $S(c)$ with center $c$ lies on $[a,b]\subset \ell$,
we compute the radius of the inner square $S'(c)$ by finding the farthest point from center $c$ lying inside $S(c)$. 
For this, we plot $L_\infty$ distances from $c$ to each point $p$ in $R$ along $[a,b]$.
Bae et al.~\cite{bbm-mwesra-2021} showed that 
the upper envelope of the $L_\infty$ distances has $O(n)$ complexity and
can be computed in $O(n)$ time.
From the centers of the annuli obtained by mapping the breakpoints of the upper envelope,
we identify those annuli which are RBSA. Finally our algorithm outputs the one with the largest width.
Considering all choices of $p_{i}$ and $p_{j}$ and the
above discussion leads to the following result.

\begin{lemma}\label{lem:emp_sq_center2}
A maximum-width RBSA corresponding to C$_{3}$ configuration can be computed 
in $O(n^3)$ time and $O(n)$ space.
\end{lemma}
\begin{proof}
For a fixed $p_{i}$ and $p_{j}$, the centers of all possible empty annuli can be computed in $O(n)$ 
time~\cite{bbm-mwesra-2021} and $O(n)$ space.  In another linear scan we can find the RBSA with the 
maximum-width. The result follows by taking all values of $i$ and $j$.
\end{proof}

We conclude the section with the following theorem.
\begin{theorem} \label{thm:cssq}
Given $n$ points in $\mathbb{R}^2$ where each one is colored with one of the $k$ colors, a maximum-width RBSA can be computed
 in $O(n^3)$ time using $O(n)$ space.
\end{theorem}

\begin{proof}
	A maximum-width RBSA results from any of the three configurations 
from Lemma~\ref{obs:sq_conf}. 
For C$_{1}$ configuration, the maximum one is reported in $O(n)$
time using $O(n)$ space from Lemma~\ref{lem:cs_emp_sq} if the points are already sorted. 
For C$_{2}$ configuration, we find a maximum-width RBLC in $O(n\log n)$ time 
using $O(n)$ space from Theorem~\ref{thm:cssq2}.
For C$_{3}$ configuration, we compute a maximum-width RBSA in $O(n^3)$ time 
using $O(n)$ space from Lemma~\ref{lem:emp_sq_center2}.
 Hence the statement followed.
 \end{proof}
\section{Maximum-Width Rainbow-Bisecting Empty Rectangular Annulus} \label{sec:rec}

In this section, we discuss an algorithm to compute a maximum-width rainbow-bisecting empty rectangular annulus
amidst a given point set $P$ on $\mathbb{R}^2$. 
The following observation shows the existence of a maximum-width RBRA meeting certain conditions.
\begin{observation} \label{obs:rec_conf1}
There exists a maximum-width RBRA \ann~with outer rectangle $R_{out}$ and inner rectangle $R_{in}$ satisfying the following conditions: (1)
Each side of $R_{out}$ contains a point of $P$ or lies  at infinity; (2) Each side of $R_{in}$ contains a point of $P$.
\end{observation}
\begin{proof}
Consider any maximum-width RBRA \ann, and assume that 
the boundary of $R_{out}$ does not contain any input point.
We enlarge each side of $R_{out}$ till each of them contains a point of $P$ or they are at infinity. 
Similarly we shrink $R_{in}$ by pushing each side inwards until each side hits a point.
Note that any two adjacent sides of $R_{in}$ can share a point at its corner.
In this process the width of
\ann~is not decreased. By our construction, we create a new maximum-width RBRA from 
\ann~that satisfies the conditions.
 \end{proof}

From Observation~\ref{obs:rec_conf1}, 
	we can find a maximum-width RBRA by examining every rectangular annulus 
	that satisfying the conditions in Observation~\ref{obs:rec_conf1}.
	The outer rectangle $R_{out}$ of a rectangular annulus $A$ can be determined by at most $4$ points of $P$, 
	and the inner rectangle $R_{in}$ of $A$ becomes the minimum rectangle that contains the points $P\cap R_{out}$.
	This results in a trivial algorithm.
\begin{remark}
For a given set of $n$ points $P$, where each point in $P$ is assigned a color from given $k$ colors, 
a maximum-width RBRA can be reported in $O(n^5)$ time. 
\end{remark}
A RBRA is called RBRA \emph{with uniform width}, or simply \emph{uniform} RBRA,
when it has all four widths (top, bottom, left and right) equal.
The following observation shows that we can construct 
a maximum width RBRA with uniform width from 
any maximum width RBRA.
\begin{observation} \label{obs:rbra_conf_uniform}
For any maximum-width RBRA, 
we can construct a maximum-width RBRA with uniform width
 satisfying the following characteristic:
 each side of its outer rectangle either contains a point of $P$
 or lies at infinity, and
 at least one side of its inner rectangle contains a point of $P$.
\end{observation}
\begin{proof}
For a maximum-width RBRA \ann, let \ann$'$ be a maximum-width RBRA that 
satisfies the conditions in Observation~\ref{obs:rec_conf1} created by transforming
\ann~according to the process in Observation~\ref{obs:rec_conf1}.
Let $R_{out}$ and $R_{in}$ be the outer and inner rectangles of \ann$'$,
and $w$ be the width of \ann$'$.
Therefore every side of $R_{out}$ contains at least one point of $P$ or lies at infinity. 
Now think that the width $w$ comes from 
the horizontal distance of left sides (left-width) of \ann$'$.
We enlarge the remaining three sides of $R_{in}$ to form another rectangle $R'_{in}$, where $R'_{in}\subseteq R_{out}$, and all the four widths (top, bottom, left and right) of the new annulus, \ann$''$, thus formed are equal. The width of \ann$''$ is equal to $w$. 
Since $R_{in} \subseteq R'_{in}$ and thus \ann$'' \subseteq$\ann$'$,
suggests \ann$''$ is also a RBRA.
See \figurename~\ref{fig:uni-ra} for an illustration.
\end{proof}
  \begin{figure}[t]
 \centering
 \includegraphics[width=.88\textwidth]{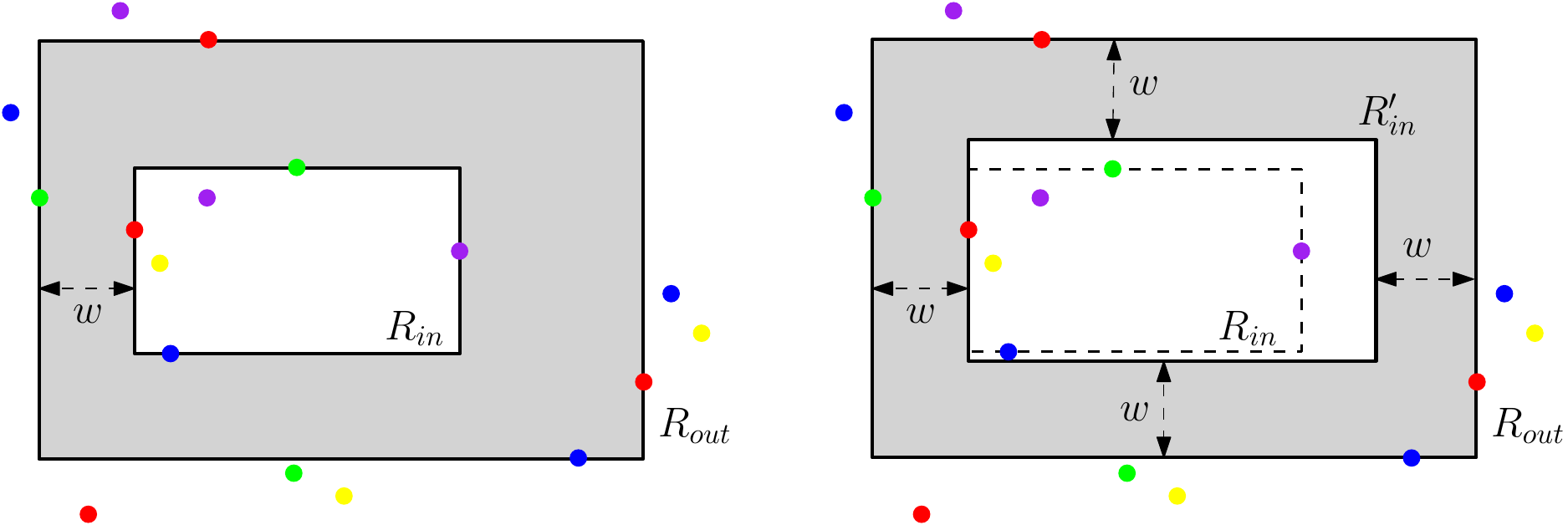}
 \caption{(left) A maximum-width RBRA of width $w$
 (right) A left-anchored maximum-width RBRA with uniform width.}
 \label{fig:uni-ra}
 \end{figure} 

As introduced in Bae et al.~\cite{bbm-mwesra-2021}, 
we call a rectangular annulus \ann~\emph{top-anchored} 
(or, bottom-anchorsed, left-anchored, right-anchored) if 
the following conditions are satisfied: 
(1) top (or, bottom, left, right, resp.) side of the outer rectangle 
 of \ann~contain a point of $P$ or lies at infinity; 
(2) top (or, bottom, left, right, resp.) side of the inner rectangle of
\ann~contains a point of $P$.
A maximum-width RBRA with uniform width has the following property.
\begin{observation} \label{obs:rec_conf}
A maximum-width RBRA with uniform width can be either 
top-anchored, bottom-anchored, left-anchored or right-anchored 
(Figure~\ref{fig:uni-ra}(right)).
\end{observation}


In a nutshell our problem boils down to the problem of computing 
an anchored and uniform maximum-width RBRA.
Here, we only discuss the case where the RBRA is ``top-anchored''. The other three cases can be similarly handled.

A brief outline of our algorithm that computes a top-anchored and uniform RBRA 
of maximum-width is as follows:
 Consider the given set of points $P = \{p_1, p_2, \ldots, p_n\}$, sorted 
 in decreasing order of their 
 $y$-coordinates and $X$ be an array of size $k$ where $X[i]$ indicates the number 
of points of color $i$. 
Consider any 
top-anchored and uniform RBRA \ann~that satisfies 
the condition of Observation~\ref{obs:rbra_conf_uniform}.
Let $p_i \in P$ be the point lying on the top side of the outer rectangle of \ann.
From the assumption on \ann,
either the bottom side of the outer rectangle of \ann~is at infinity or
there is another point $p_j \in P$ for $i < j \leq n$ on it.
Since \ann~is top-anchored, there is a third point $p_k \in P$ on the top side of the inner rectangle of \ann.
Note that the width of \ann~is determined by the $y$-difference of $p_i$ and $p_k$, 
that is, $y(p_i) - y(p_k)$.
Thus, the maximum width for top-anchored annuli is one among $O(n^2)$ values $\{y(p_i) - y(p_k) \mid 1 \leq i \leq k \leq n\}$. \\
Initially we study the case where two points $p_i$ and $p_j$ on the top and bottom sides are fixed,
and then discuss the situation where only a point $p_i$ on the top side is fixed.
We discuss a decision algorithm when two points on the top and bottom sides of the outer rectangle are fixed, and use it as a sub-routine for solving the other case.
 
\subsection{Decision problem for fixed top and bottom sides}\label{sec:dec-rbra}
Here we discuss a decision algorithm when two points $p_i$ and $p_j$ are fixed on the top and bottom sides of outer rectangle of the RBRA, where $2\leq i+1 < j \leq n$.
To represent the case when the bottom side of the outer rectangle lies at infinity, 
we use ${j}=\infty$.
The decision problem $DP_{ij}(w)$ is defined as follows:
\begin{center}
\noindent\framebox{\begin{minipage}{5.5in}	
\emph{Given}: A positive real $w>0$, 
two indices $i$ and $j$ where  $2\leq i+1< j\leq n$, or $j=\infty$.\\
\emph{Task}: Decides the existence of a RBRA of width at least $w$ and whose outer rectangle contains $p_{i}$ and $p_{j}$ on its top and bottom sides, respectively.
\end{minipage}}
\end{center}
The following observation holds on $DP_{ij}(w)$.
\begin{observation} \label{obs:dera}
 If $DP_{ij}(w)$ is TRUE, then $DP_{ij}(w')$ is TRUE for any $w'\leq w$.
 On the other hand, if $DP_{ij}(w)$ is FALSE, then $DP_{ij}(w')$ is FALSE for any $w'\geq w$.
\end{observation}
To compute $DP_{ij}(w)$, we consider the set $P_{ij}:=\{p_{i+1}, \ldots, p_{j-1}\}$
containing points between $y(p_{i})$ and $y(p_{j})$
where $2 \leq i+1 < j\leq n$, and $P_{i\infty} := \{p_{i+1}, \ldots, p_{n}\}$. 
The decision algorithm perform the \emph{$y$-range $x$-neighbor query} operation~\cite{bbm-mwesra-2021} on points of $P_{ij}$ to evaluate 
$DP_{ij}(w)$.\\
We present Algorithm~\ref{alg:rbra-ours} to answer the decision problem $DP_{ij}(w)$.
The algorithm calls the \emph{$y$-range $x$-neighbor query} twice. Also,
if the algorithm decides that $DP_{ij}(w)$ is TRUE,
then it also returns a corresponding RBRA,
that is, a RBRA with uniform width $w$ such that
$p_i$ and $p_j$ on the top and bottom sides of its outer rectangle.

\begin{algorithm}[H]
\footnotesize
\KwIn{A positive real $w>0$, two indices $i$ and $j$ where $2\leq i+1< j\leq n$, or $j=\infty$.}
\KwOut{$DP_{ij}(w)$ and if $DP_{ij}(w)$ is TRUE then the algorithm outputs a 
	RBRA with uniform width $w$ having $p_{i}$ and $p_{j}$ lying on the top and bottom sides of its outer rectangle, respectively.}

\If {$y(p_{i}) - y(p_{j}) < 2w$}
 {Return FALSE.}
 
 Search the rightmost and the leftmost points from the sets 
 $P_{ij} \cap [-\infty, x(p_{i})] \times [y(p_{i}) - w, y(p_{i})]$ and
 $P_{ij} \cap [x(p_{i}), +\infty] \times [y(p_{i}) - w, y(p_{i})]$ respectively. Let these two points be
 $l_1$ and $r_1$ respectively.\\
 
 Search the rightmost and the leftmost points from the sets 
 $P_{ij} \cap [-\infty, x(p_{j})] \times [y(p_{j}), y(p_{j})+w]$ and
 $P_{ij} \cap [x(p_{j}), +\infty] \times [y(p_{j}), y(p_{j})+w]$ respectively. Let these two points be
 $l_2$ and $r_2$ respectively. \\

 Let $p_l$ be the rightmost one in $\{l_1, l_2\}$ and
 $p_r$ be the leftmost one in $\{r_1, r_2\}$.\\
 
\If{$\min\{x(p_{i}), x(p_{j})\} < x(p_l)< \max\{x(p_{i}), x(p_{j})\}$ or 
$\min\{x(p_{i}), x(p_{j})\}<x(p_r)<\max\{x(p_{i}), x(p_{j})\}$}
 {Return FALSE.}

Search all the horizontal gaps of length at
least $w$ from the points lying in the
range $(x(p_l), \min\{x(p_{i}), x(p_{j})\}+w)$. Let $\mathcal{L}$ denote the set of such gaps.  \tcp{Boundary points are included in the search range.}
 
Search all the horizontal gaps of length at least $w$ from the points lying in the
range $(\max\{x(p_{i}), x(p_{j})\}-w, x(p_r))$. Let $\mathcal{R}$ denote the set of such gaps.  \tcp{Boundary points are included in the search range.}

 \If {$\mathcal{L}\neq \emptyset$ and $\mathcal{R}\neq \emptyset$}
 {Start with the leftmost horizontal gap (say, $g_l:=x(l_a)-x(l_b)$) from set $\mathcal{L}$. Maintain the number of points of each color on the left side of $g_l$. From $x(l_a)$ start moving forward until a point say, $p_{in}$ comes such that the points from $x(l_a)$ to $x(p_{in})$ forms a rainbow region. Move forward from $p_{in}$ until a point say, $p_{out}$ is reached such that if moved further the region outside the region bounded by 
 $x(l_b)$, $x(p_{out})$, $y(p_{i})$ and $y(p_{j})$ is not rainbow. Search a $w$ width horizontal gap between $x(p_{in})$ and $x(p_{out})$ and check if it lies within $(\max\{x(p_{i}), x(p_{j})\}-w, x(p_r))$.
 Repeat the procedure for all the gaps in $\mathcal{L}$ until such a gap is found that defines the right sides of RBRA. Return TRUE if such gap exists, and RBRA of uniform width $w$. Otherwise Return FALSE.}
 \Else{Return FALSE.}
\caption{Decision algorithm}	
\label{alg:rbra-ours}
\end{algorithm}

Next, we prove the correctness of Algorithm~\ref{alg:rbra-ours}.
\begin{lemma}\label{lem:rb_rect1}
Algorithm~\ref{alg:rbra-ours} correctly computes $DP_{ij}(w)$ for any given $w > 0$ in $O(n)$ time.
\end{lemma} 
 \begin{proof}
The time complexity of Algorithm~\ref{alg:rbra-ours} depends on the search queries given in line numbers 3, 4, 8, 9 and 11. In line 3 and 4, the search queries can be done in $O(\log n)$ 
time~\cite{bbm-mwesra-2021}.
Since the points are maintained in sorted order, the operations in lines 8, 9 and 11 can be done in 
$O(n)$ time.
For any horizontal gap of length at least $w$ from $\mathcal{L}$ that defines the left sides of the RBRA we search for the corresponding right sides such that rainbow property is satisfied and right sides are in 
$\mathcal{R}$.
In line 11, we start searching for points $p_{in}$ and $p_{out}$ for the leftmost gap $g_l$ in set 
$\mathcal{L}$. At the same time we maintain a count on the number of points of each color we seen so far, starting from $x(l_a)$. At 
$x(p_{in})$ we have exactly $k$ colors present. Similarly at $x(p_{out})$ we have exactly $k$ colors 
present outside the region (including the boundaries) bounded by lines $x(l_b)$, $x(p_{out})$, $y(p_{i})$ and $y(p_{j})$ from left, right, top and bottom respectively. This operation is repeated for all gaps in 
$\mathcal{L}$ and every time it advances to the next gap, 
the points $p_{in}$ and $p_{out}$ move forward from their previous positions. 
During the repetition, if Algorithm~\ref{alg:rbra-ours} finds a RBRA with uniform width $w$, 
it returns TRUE and the RBRA. 
\end{proof}

\subsection{Optimization Algorithm}\label{sec:opt-rbra}
We now move to the case where we fix only a point $p_{i} \in P$ that defines the top side of the outer rectangle and compute a top-anchored and uniform RBRA~\ann$'$
of maximum-width for the point set $P$. 

If $p_{i}$ lies on the top side of the outer rectangle of a
top-anchored and uniform RBRA \ann,
then the width of \ann~should be in the set $W = \{y(p_{i}) - y(p_{k}) \mid i < k \leq n\}$. 
Thus, we can find \ann$'$ by solving $DP_{ij}(w)$ with Algorithm~\ref{alg:rbra-ours}
for every $i+1<j$ and $w\in W$.
There are a total of $O(n^2)$ of $DP_{ij}(w)$ to solve, 
but we can find \ann$'$ by solving only $O(n)$ of them.
\begin{lemma} \label{lem:rb_rect2}
 For a point $p_i\in P$, a top-anchored and uniform RBRA~\ann of maximum-width such that 
$p_i$ lies on the top side of the outer rectangle of \ann~can be computed in $O(n^2)$ time.
\end{lemma}
\begin{proof}
For a fixed $p_{i}$, there are $O(n)$ candidate widths $w$ of the annulus and 
$O(n)$ candidate points $p_j$ which lie on the bottom side of the outer rectangle.
We need to find the maximum width $w^*$ such that $DP_{ij^*}(w^*)$ is TRUE 
for an index $j^*$ with $i+2< j^*$.

Let index $j=i+2$ and width $w=y(p_i) - y(p_{i+1})$ 
be the smallest possible values among the candidates, 
and solve $DP_{ij}(w)$ using Algorithm~\ref{alg:rbra-ours}.
If $DP_{ij}(w)$ is TRUE, then we set $w$ as the next smallest candidate width $w=y(p_i) - y(p_{i+2})$
and solve $DP_{ij}(w)$ again.
We repeat this process until $DP_{ij}(w)$ becomes FALSE.
From Observation~\ref{obs:dera}, $DP_{ij}(w')$ is FALSE for all $w'>w$, so 
we don't need to solve $DP_{ij}(w')$ for $w'>w$.
Thus we increase $j$ by one, and repeat the process to find 
the maximum width such that $DP_{ij}(w)$ is TRUE.
We repeat the entire process until $j=\infty$ or $w=y(p_{i}) - y(p_{n})$,
and we can find the maximum width $w'$ during the process.

The suggested process only increase both $w$ and $j$, 
	so we call Algorithm~\ref{alg:rbra-ours} at most
$O(n)$ times. Therefore it takes $O(n^2)$ time to report 
a top-anchored and uniform RBRA of maximum-width for $p_{i}$.
\end{proof} 

We have $n$ possible choices for $p_i$, so we can find 
a top-anchored and uniform RBRA of maximum-width in $O(n^3)$ time, and the
other three cases (bottom-anchored, left-anchored, or right-anchored)
of the maximum-width RBRA are handled similarly. 
We conclude the section with the following result. 
\begin{theorem} \label{thm:rect}
 Given a set $P$ of $n$ points in the plane where each one is assigned a color from given $k$ colors,
 a maximum-width RBRA with respect to $P$
 can be computed in $O(n^3)$ time and $O(n)$ space.
\end{theorem} 

\subsection{An Improved Approach}\label{sec:modified}

In Section~\ref{sec:dec-rbra}, we present a $O(n)$ time decision algorithm to compute a maximum-width RBRA, where the top and bottom sides of its outer rectangle pass through two fixed points. Here we show that Algorithm~\ref{alg:rbra-ours} can be implemented in $O(k^2\log n)$ time.\\
Let $a$ be any $w$-gap from the set $\mathcal{L}$ (described in Step 8 of Algorithm~\ref{alg:rbra-ours}). We are interested in the leftmost $w$-gap $b$ from the set $\mathcal{R}$ such that the corresponding inner rectangle is rainbow and the number of points outside of the annulus is maximum. Note that the inner rectangle defined by $a$ and any $w$-gap $b'$ in $\mathcal{R}$ to the right of $b$ is rainbow, while the number of points is less lying outside of the outer rectangle of the annulus. On the other hand, if we find the rightmost $w$-gap $a'$ in 
$\mathcal{L}$ such that the corresponding inner rectangle defined by $a'$ and $b$ is rainbow, then either $a=a'$ or $a'$ lies to the right of $a$. Note that in the latter case we can ignore the $w$-gap $a$, since the annulus defined by $w$-gaps ($a', b$) is better than that defined by ($a, b$) in the sense of the number of points lying outside the outer rectangle of the annulus.

\subsubsection{Minimal Rainbow Intervals}\label{sec:mri}
Consider the set $L$ of points $P_{ij} \cap [x(p_l), \min\{x(p_{i}), x(p_{j})\}+w] \times [y(p_{j}), y(p_{i})]$. Then, $\mathcal{L}$ denotes the set of $w$-gaps in $x$-coordinates of $L$. 
Similarly $R$ be the set of points 
$P_{ij}\cap[\max\{x(p_{i}), x(p_{j})\}-w, x(p_r)] \times [y(p_{j}), y(p_{i})]$. 
For our purpose we only consider the 
$x$-coordinates of the points in $P_{ij}$ and assume that the points are projected down on the $x$-axis.\\
Let $a$ be any point in $L$. We define $r(a)$ as the leftmost point in $R$ such that 
$P_{ij}\cap [a, r(a)]$ is rainbow. Similarly for any point $b$ in $R$, define $l(b)$ as the rightmost point in $L$ such that $P_{ij}\cap[l(b), b]$ is rainbow. Note that as $a$ moves to the right, 
$r(a)$ either moves to the right or stays at the same position. Similarly it holds for $l(b)$. We call an interval $[a, b]$ as \emph{minimal rainbow interval} with $a \in L$ and $b \in R$, if $b=r(a)=r(l(b))$ and 
$a=l(b)=l(r(a))$. In the following we give some properties of minimal rainbow interval.\\
Recall that $\alpha(p)$ denotes the color of a point $p$, and let $c_{[Q]}$ denotes the number of points of color $c$ in a point set $Q$.\\
\begin{itemize}
 \item[(i)] If $[a, b]$ is a minimal rainbow interval, then either $\alpha(a)_{[a, b]}= 1$ or $a$ is the rightmost point in $L$; $\alpha(b)_{[a, b]}= 1$ or $b$ is the leftmost point in $R$. In other words, $a$ is the only point of that particular color in $[a, b]$ unless $a$ is the rightmost point in $L$, and so is $b$.
 \item[(ii)] No two minimal rainbow intervals are nested, and any two minimal rainbow intervals overlap.
 \item[(iii)] There cannot be two minimal rainbow intervals $[a_1, b_1]$ and $[a_2, b_2]$ such that $a_1$ and $a_2$ are in the same color, or $b_1$ and $b_2$ are in the same color. Also from (ii) we can order all minimal rainbow intervals from left to right, $[a_1, b_1]$, $[a_2, b_2]$,\ldots, $[a_m, b_m]$ with $a_1 < a_2 <\ldots< a_m$ and $b_1 < b_2 <\ldots< b_m$. Since the colors of $a_1, a_2\ldots, a_m$ (resp. $b_1, b_2\ldots, b_m$) must be all distinct, implies that $m$ can be at most $k$. 
\end{itemize}
We are interested in the $w$-gaps that are generated in between minimal rainbow intervals.
In the following lemma we present a bound on the number of ``relevant'' $w$-gaps.
More specifically, for a given $p_i$ and $p_j$, a $w$-gap is relevant if $w \in \mathcal{L}$ or $w \in \mathcal{R}$.
\begin{lemma}\label{lem:rel-w-gaps}
The number of relevant $w$-gaps is at most $O(k)$.
\end{lemma}
\begin{proof}
 Consider any $w$-gap $A$ in $\mathcal{L}$. Also assume that $A$ is between the left endpoints of minimal rainbow intervals $[a_1, b_1]$ and $[a_2, b_2]$, i.e., $A$ lies between $a_1$ and $a_2$. Therefore, the leftmost $w$-gap $b$ in $\mathcal{R}$ such that the inner rectangle defined by ($A, B$) is rainbow must lies somewhere to the right of $b_2$. If there lies another $w$-gap $A'$ between $a_1$ and $a_2$ that is to the right of $A$, we can simply ignore $A$. In a nutshell, we consider only the rightmost $w$-gap to the left of the left endpoint of each minimal rainbow interval and the leftmost $w$-gap to the right of the right endpoint of each minimal rainbow interval. Hence, the number of these ``relevant'' $w$-gaps is at most $2k$.
\end{proof}

\subsubsection{Modified Decision Algorithm}\label{sec:new-da}
Here we discuss the new decision problem $DA_{new}$ that computes $DP_{ij}(w)$.
$DA_{new}$ have exactly same lines 1 to 7 as stated in Algorithm~\ref{alg:rbra-ours}. The new algorithm then computes all possible minimal rainbow intervals in between $y(p_i)$ and $y(p_j)$ as follows. Consider the points in $L$. For each point $a$ in $L$, we compute $r(a)$ as follows. For each color $c\in k$, compute the closest point in $P_{ij}$ of color $c$ to the right of $a$ and then find the farthest point $b$ from them. If $b$ is in $R$, then $r(a)=b$; otherwise $r(a)$ is the leftmost point in $R$ by definition. Similarly, for each point $b$ in $R$, we compute $l(b)$.
\begin{figure}[t]
	\centering
	\includegraphics[width=.45\textwidth]{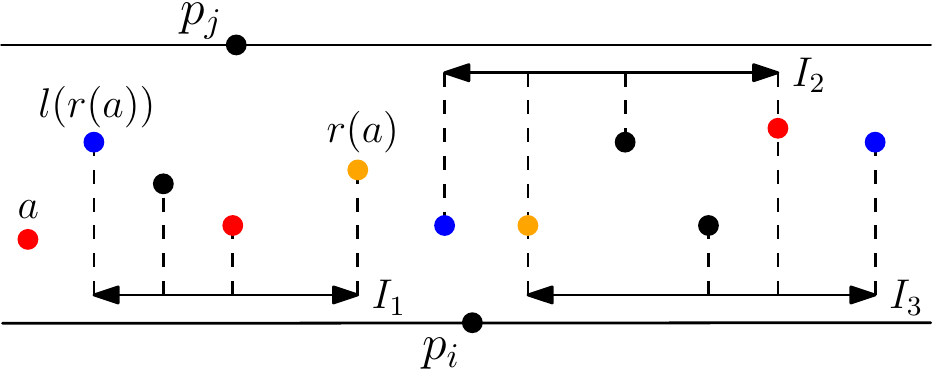}
	\caption{For the leftmost point $a$,  $r(a)$ and $l(r(a))$ define the leftmost minimal rainbow interval $I_1$. Two other minimal rainbow intervals are denoted by $I_2$ and $I_3$.}
	\label{fig:rainbowinterval}
\end{figure}

To compute all minimal rainbow intervals, first we start with the leftmost point $a\in L$, compute $b=r(a)$, and
then compute $l(b)$. Therefore, $[l(b), b]=[a_1, b_1]$ denotes the first minimal rainbow interval. 

See Figure~\ref{fig:rainbowinterval} as an example.
If $a_1$ is the rightmost point in $L$, then $[a_1, b_1]$ is the only minimal rainbow interval. Otherwise, we search for the second minimal rainbow interval $[a_2, b_2]$. Since there is no point of color 
$\alpha(a_1)$ in $[a_1, b_1]$ other than $a_1$, 
therefore $\alpha(b_2)=\alpha(a_1)$. 
So, $b_2$ can be easily found by finding the closest point of color $\alpha(a_1)$ 
in $P_{ij}$ to the right of $b_1$. Also then, $a_2=l(b_2)$. The above process continues until we reach the rightmost point in $L$.

Consider any minimal rainbow interval $[a, b]$. We store $k$ counters for $[a, b]$. 
Let $C([a, b])=\#\_1([a, b]), \#\_2([a, b]),\ldots, \#\_k([a, b])$ be the color counter vector, consisting of the number of points in $P_{ij}$ of each color. 
 For each minimal rainbow interval, we store its color counter vector.\\
For each $[a, b] \in [a_1, b_1]$, $[a_2, b_2]$,\ldots, $[a_m, b_m]$, we repeat the following steps: \\\\
$(i)$ First we compute the rightmost $w$-gap $A$ to the left of $a$, if any. Find the leftmost $w$-gap $B$ to the right of $b$, if any. \\
$(ii)$ Compute the color counter vector of the region between $A$ and $a$ and of the region between $b$ and $B$.
 Also compute the counter $C([A, B])$ for points between $A$ and $B$.\\
$(iii)$ From $C([A, B])$, the number of points of each color that lie outside the annulus can be found.\\\\
After the completion of the above steps if there exixts a RBRA of $w$-width, $DA_{new}$ reports it.
In order to perform the above operations efficiently the following data structures are being used.
\begin{itemize}
 \item Consider the 1-dimensional range tree $\mathcal{X}_{ij}(c)$, 
 where $c \in [k]$, for the $x$-coordinates
 of points in $P_{ij}$ for counting for color $c$. Also at each node $v$ of the tree, we store $size(v)$, the number of leaves descended from $v$.  The data structure $\mathcal{X}_{ij}(c)$ can be constructed
 using storage $O(|P_{ij}(c)|)$~\cite{BCKO08}, where $|P_{ij}(c)|$ represents the number of points of color $c$ in $P_{ij}$. We maintain $O(k)$ such 1D range trees.
 \item We preprocess the points of $P_{ij}$ into a 2D range tree (with fractional cascading) $\mathcal{T}_{ij}$ as follows: Let $x_1<x_2<\ldots$, be the $x$-coordinates of points in $P_{ij}$ in sorted order. Here we map each $x_i$ into a 2D point ($x_i$, $x_{i+1}-x_i$). We construct $\mathcal{T}_{ij}$ on these 2D points.
 $\mathcal{T}_{ij}$ can be constructed using storage $O(|P_{ij}|\log |P_{ij}|)$~\cite{BCKO08}.
\end{itemize} 
With the above data structures in hand we can compute all minimal rainbow intervals and perform the mentioned operations on them as follows:
\begin{itemize}
\item For any point $a \in L$ (resp. for each point $b$ in $R$), we find $r(a)$ (resp. $l(b)$) in $O(k \log n)$ time from $k$ given 1D range trees. Thus in $O(k \log n)$ time we compute a minimal rainbow interval.
\item To compute the leftmost $w$-gap to the right of a minimum rainbow interval $[a, b]$, 
more specifically the leftmost $w$-gap to the right of $b$,
we use a 3-sided range query defined by the region $R=[x(b), min(x(r_1), x(r_2))-w] \times [w, \infty]$, 
where $r_1$ is the leftmost point from the region $[x(p_i), \infty] \times [y(p_i)-w, y(p_i)]$,
and $r_2$ is the leftmost point from the region $[x(p_j), \infty] \times [y(p_j), y(p_j)+w]$. 
Here we find the leftmost point from region $R$. From $\mathcal{T}_{ij}$ we can find the leftmost 
point in $R$ in $O(\log n)$ time. Any $w$-gap corresponds to a point in $R$ and can be used to 
construct an empty annulus of width $w$.
\item  For any rectangular region $R$, the color counter vector $C(R)$ of $R$ can be computed 
in $O(k\log n)$ time from $k$ given 1D range trees. 

Similarly, the color counter vector of the region 
between $A$ and $a$ (resp. between $b$ and $B$) and $C([A, B])$ (mentioned in step $(ii)$) 
can be computed in $O(k\log n)$ time.
\item Using $C([A, B])$ we can find the number of points outside the empty annulus in $O(k)$ time.
 \end{itemize}

\begin{lemma}\label{lem:rbra-newdec}
$DA_{new}$ computes $DP_{ij}(w)$ for any given $w > 0$ in $O(k^2\log n)$ time using $O(n \log n)$ space.
\end{lemma}
\begin{proof}
$DA_{new}$ takes $O(\log n)$ time~\cite{bbm-mwesra-2021} to compute lines 3 and 4. After getting TRUE in line 7 of Algorithm~\ref{alg:rbra-ours},  $DA_{new}$ computes all minimal rainbow intervals. Since there are $O(k)$ minimal rainbow intervals in total, it takes $O(k^2 \log n)$ time.
For any minimal rainbow interval $DA_{new}$ construct a $w$-width RBRA if exists in $O(k \log n)$ time and 
$O(n \log n)$ space using the data structures mentioned above. Considering all minimal rainbow intervals the total running time is $O(k^2 \log n)$.
\end{proof}

Now we discuss how to report a maximum-width top-anchored RBRA when a point $p_i \in P$ is fixed on the top side of outer rectangle of RBRA.
Assume any point $p_i \in P$ lies on the top side of outer rectangle of RBRA.
For a fixed $p_i$, we consider the widths that lies in the set $W_i := \{y(p_i) - y(p_k) \mid i \leq k \leq n\}$ and report the maximum among them. The optimization algorithm is described in section~\ref{sec:opt-rbra}.
Having $p_i$ fixed, consider all possible values of $j = i+2, \ldots, n$ and also $j = \infty$. 
As $j$ increases, more points are included in $P_{ij}$ and therefore we update the structures by inserting a point in each of them. 
Therefore in this case we can report a maximum width top-anchored RBRA in $O(k^2 n \log n)$ time.
Considering all values of $p_i$, along with the above discussion we have the following result.

\begin{theorem}\label{theo:new-op-algo}
 For a given set of $n$ points in the plane where each point is assigned a color from given $k$ colors, a maximum-width top-anchored RBRA can be computed in $O(k^2n^2\log n)$ time and $O(n\log n)$ space. 
\end{theorem}

\section{Maximum-width Rainbow-Bisecting Empty Circular Annulus}\label{sec:cir}
In this section, we compute a maximum-width RBCA from a given point set $P$ on $\mathbb{R}^2$. Here we assume no four input points are concyclic. 
We begin with the geometric characterizations of a maximum-width RBCA. 

\begin{observation}
 The boundaries of outer circle and inner circle of a maximum-width RBCA \ann~contain at least one input 
 point from $P$.
\end{observation}
We restrict our target annuli whose inner circles have more than one point. Without this restriction
our problem can be solved in $O(n\log n)$ time~\cite{dhmrs-leap-03}. Also we define a maximum-width RBCA as ``special'' if there is a point on the inner boundary that lies on the straight line joining the center of the annulus and the point lying on the outer boundary. The following lemma introduced is similar to Lemma 2.1 and Lemma 2.2 
in~\cite{dhmrs-leap-03}.
\begin{figure}[t]
	\centering
	\includegraphics[width=.41\textwidth]{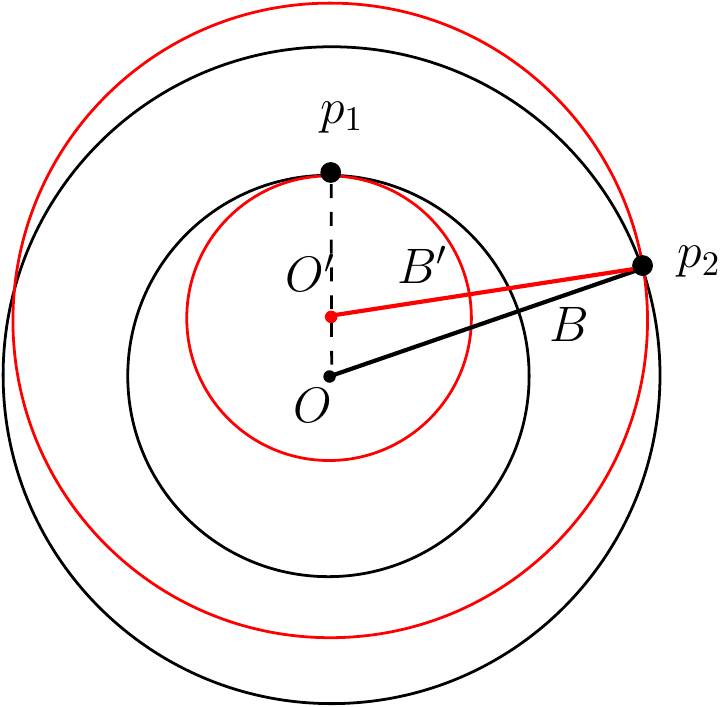}
	\caption{A RBCA \ann~(in black)that is not special. Another RBCA $\mathcal A'$ (in red) is constructed from \ann~and 
	have larger width.}
	\label{fig:circannulus1}
\end{figure}

\begin{lemma}\label{lem:un_cir_annulus}
    A maximum-width RBCA \ann~can have one of the following potential configurations - $(i)$ Both the boundaries of $C_\mathrm{out}$ and $C_\mathrm{in}$ contain at least two input points of $P$ 
	when \ann~is not a special annulus and 
	$(ii)$ The boundary of $C_\mathrm{out}$ contain at least one input point of $P$ and the boundary of 
	$C_\mathrm{in}$ contain at least two input points of $P$ when \ann~is a special annulus.
\end{lemma}

\begin{proof}
For $(i)$, consider \ann~ be any maximum-width RBCA which is not special and contains a single input 
point, say $p_1$ 
on the boundary of the inner circle $C_\mathrm{in}$ and the boundary of  the outer circle $C_\mathrm{out}$ contains at least one point 
of $P$, say $p_2$ (Ref. Figure~\ref{fig:circannulus1}). If we move the center (at $O$) of \ann~along the line joining the center
and $p_1$ as well as maintaining the rainbow  property on both sides of the annulus we can construct another RBCA $\mathcal A'$ (red in color in Figure~\ref{fig:circannulus1}, center at $O'$) whose boundaries still contain $p_1$ and $p_2$ . Note $Op_2$ intersects the initial inner circle (i.e., black in color in Figure~\ref{fig:circannulus1} ) at $B$ and similarly $O'p_2$ intersects the transformed inner circle (i.e., red in color in Figure~\ref{fig:circannulus1}) at $B'$.

Let us denote the width of the initial annulus \ann~and the transformed annulus $\mathcal A'$ by $w_{bl}(\mathcal A)$ and $w_{red}(\mathcal A')$ respectively, where $w_{bl}(\mathcal A)=Op_2-OB$ and $w_{red}(\mathcal A')=O'p_2-O'B'$. Then the difference between the width of the transformed and the  initial annulus i.e $w_{red}(\mathcal A')-w_{bl}(\mathcal A)=(O'p_2-O'B')-(Op_2-OB)$. Applying triangle inequality, we get  $(O'p_2-O'B')-(Op_2-OB)\geq Op_2-O'O-O'B'-Op_2+OB\geq(OB-O'O)-O'B'\geq O'B-O'B'\geq 0$. This proves that the width of the transformed annulus $\mathcal A'$  is greater than the initial annulus \ann~ implying that  \ann~ is not an optimum annulus. The other case where the boundary of inner circle of \ann~contain at least 2 input points 
of $P$ and the boundary of its outer circle contains a single input point can be handled similarly.\\

Next, we will prove the validity of case $(ii)$ of Lemma \ref{lem:un_cir_annulus}. Let $\mathcal A$ be a maximum-width RBCA that is special and $C_{in}$ of $\mathcal A$ contains two input points of $P$. If both the points lies on the boundary of $C_{in}$ of $\mathcal A$, then it is trivially proved. If this is not the case, we can construct another RBCA $\mathcal A'$ from $\mathcal A$ like previous construction such that both $\mathcal A$ and $\mathcal A'$ have equal widths and there are two points lying on the inner boundary of $\mathcal A'$.
\end{proof}
We further denote any maximum-width RBCA having configuration $(i)$ and$(ii)$ of Lemma~\ref{lem:un_cir_annulus} as 
$Cir_\mathrm{22}$ and $Cir_\mathrm{21}$ respectively. \\
%
%
%
%
We adopt the technique described in D\`{i}az-B\'{a}\~{n}ez et al.~\cite{bls-laop-06} where they solved a problem known in the literature as \emph{obnoxious plane problem} which is defined as follows: 

Given a set $S$
containing $n$ points in $\mathbb{R}^3$, find a plane $\xi$ that intersects the convex hull of $S$ and maximizes the minimum Euclidean distance to the points. In $\mathbb{R}^2$ this problem reduces to computing
of the \emph{widest empty corridor}~\cite{mh88} through a set of points. Finally they map a corridor from 
$\mathbb{R}^2$
to a slab in $\mathbb{R}^3$. A \emph{slab} is defined as the open region of $\mathbb{R}^3$ through $S$ and is bounded by two parallel planes that intersect the convex hull of $S$ and the \emph{width} of the slab is the distance between the bounding planes. They solved the \emph{widest empty slab} problem to solve obnoxious plane problem. They have shown that the candidate planes belongs to the following categories - $(i)$ $C_\mathrm{11}$ - Each of the bounding planes of the slab contain an input point. $(ii)$ $C_\mathrm{31}$ - One bounding plane contains 3 input points and other contain a single input 
point $(iii)$ $C_\mathrm{22}$ - Each bounding plane contain 2 input points and $(iv)$ $C_\mathrm{21}$ - One bounding plane contains 2 input points and other contain a single input. \\
To apply the above method we transform the points of set $P$ from $\mathbb{R}^2$ to $\mathbb{R}^3$ using paraboloid lifting as follows:
The points of the input set $P$ are transformed from $\mathbb{R}^2$ to points in $\mathbb{R}^3$ using paraboloid transformation. Let $P'= \{(p_x, p_y, p_x^2 + p_y^2): (p_x, p_y) \in P\}$ be the set of points in 
$\mathbb{R}^3$.
Each point in $\mathbb{R}^3$ is obtained from the vertical projection of the corresponding point in 
$\mathbb{R}^2$ on the parabola $z = x^2 + y^2$. Under this transformation a circle of radius $r_1$ and center at $(c_x, c_y)$ in 
$\mathbb{R}^2$ becomes a plane $z= 2c_x.x + 2c_y.y - (c_x^2 + c_y^2 - r_1^2)$ in $\mathbb{R}^3$ and any two concentric circles are transformed to two parallel planes (Ref. Figure~\ref{fig:circannulus}). 

Consider two concentric circles of an empty annulus having radius $r_1$ and $r_2$, where $r_2$>$r_1$. The corresponding two
parallel planes in $\mathbb{R}^3$ are $z_1= 2c_x.x + 2c_y.y - (c_x^2 + c_y^2 - r_1^2)$ and
$z_2= 2c_x.x + 2c_y.y - (c_x^2 + c_y^2 - r_2^2)$ respectively. Now consider the dual representation of the points in $P'$. In this representation, a point $p\in P$ lies on (resp. inside, outside) a circle $ci$ in 
$\mathbb{R}^2$ if and only if the dual hyperplane $ci^*$
contains (respectively passes above, below) the dual point $p^*$. This implies that the maximum-width RBCA 
problem in $\mathbb{R}^2$ reduces to widest empty slab problem (where each plane is colored) 
in $\mathbb{R}^3$. In this reduction a candidate slab of $C_\mathrm{22}$ category represents a 
$Cir_\mathrm{22}$ annulus and a candidate slab of $C_\mathrm{21}$ category represents a 
$Cir_\mathrm{21}$ annulus (Figure~\ref{fig:cand_planes}). Finally, the optimal solution in our problem is obtained by tracking the slabs 
of $C_\mathrm{22}$ and $C_\mathrm{21}$ categories which span $k$ colors on both sides of it.

%
%

\begin{figure}[t]
	\centering
	\includegraphics[width=.6\textwidth]{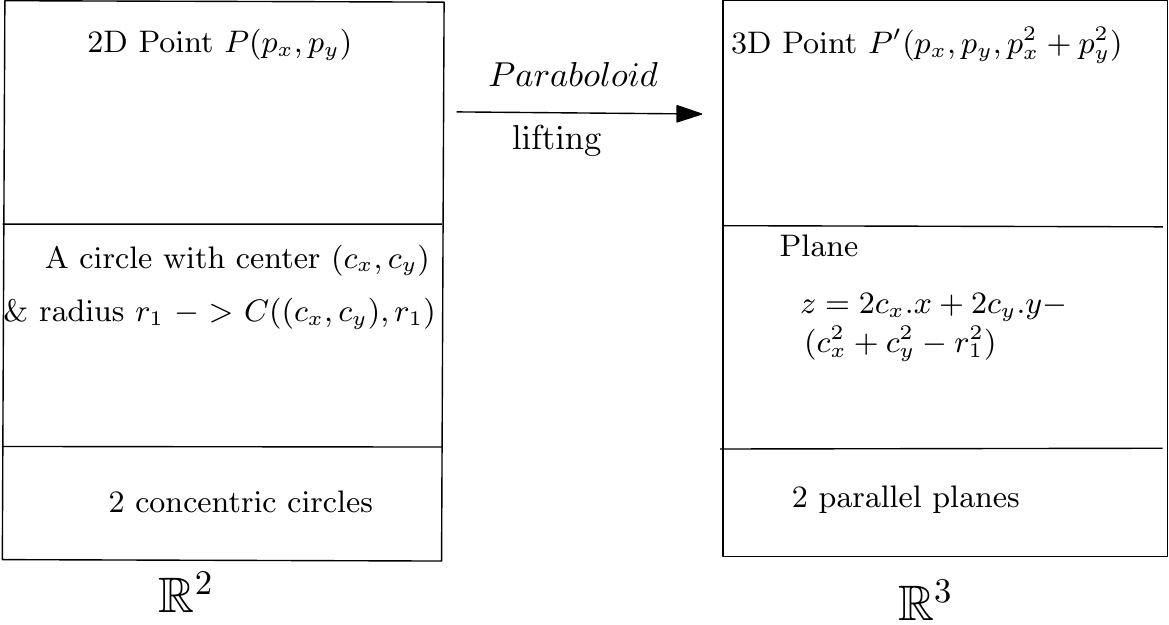}
	\caption{Paraboloid transformation of the input points. Left side maps right side.}
	\label{fig:circannulus}
\end{figure}
%

The approach of solving widest empty slab problem is primarily based on topological sweep over the arrangements of planes corresponding to the dual representation of the points in $P'$. 
Let $\mathcal{B}(H)$ be the arrangement of dual planes corresponding to points in $P'$, where each plane has a color corresponding to the point in primal. Using topological sweep on $\mathcal{B}(H)$ all candidate empty slabs of type $C_\mathrm{21}$ and $C_\mathrm{22}$ are generated. 
The topological sweep algorithm in $\mathbb{R}^3$ is described in~\cite{apg-tsitd-90} which generalizes the method for sweeping arrangements of lines in $\mathbb{R}^2$~\cite{eg-tsaa-89} to $\mathbb{R}^3$.
The data structure used to keep information about the cells intersected by the sweeping surface at any given time is illustrated in~\cite{eg-tsaa-89}.
In addition to that we maintain a count on the number of planes of each color $\in[k]$ present above and below for each edge of the current planar cut of each plane
during the sweep. For each plane $k'\in \mathcal{B}(H)$, we consider three arrays $UP_{k'}[1:k]$, 
$LOW_{k'}[1:k]$ and $CO_{k'}[1:n-1]$ as additional local data structure. Once the initial planar cut for plane $k'$ is determined, we initialize the above arrays
following sequence of edges in $N_{k'}[1:n-1]$, where $N_{k'}[1:n-1]$ 
is a list of pairs of indices indicating the lines delimiting
each edge of the current planar cut for plane $k'$.
The array $CO_{k'}[i]= (a, b)$, indicates the number of colored planes above and below each edge of $N_{k'}[i]$ of current planar cut for $k'$. We use $UP_{k'}$ and $LOW_{k'}$ to initialize $CO_{k'}$.
Consider initial cut of each plane $k'$, for $N_{k'}[1]$, we have $UP_{k'}[j]=1$, where $j$ is the color of the edge in $N_{k'}[1]$, $LOW_{k'}[j]=1$ for all $j = \{1, 2, \dots, k\}$ and $CO_{k'}[1]= (1, k)$. We find each $CO_{k'}$ in $O(n)$ time. After the sweep performs an elementary step, the incoming edges are swapped at the corresponding vertex in each planar cut and  
$CO_{k'}$ is updated following new $N_{k'}$. This operation requires $O(1)$ time.

\begin{figure}[t]
	\centering
	\includegraphics[width=.6\textwidth]{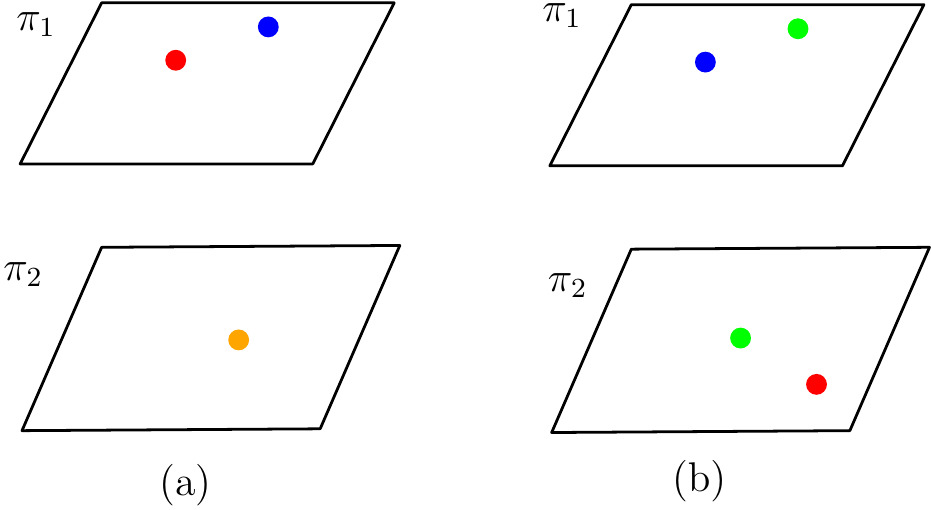}
	\caption{(a) A $C_\mathrm{21}$ empty slab defined by 2 parallel planes $\pi_{1}$ and $\pi_{2}$.
		(b) A $C_\mathrm{22}$ empty slab.}
	\label{fig:cand_planes}
\end{figure}

The above discussion along with lemma $4$~\cite{bls-laop-06} leads to the following result.
\begin{theorem}\label{thm:un_cir_annu1}
	Given a set of $n$ points in $\mathbb{R}^2$,  each one is colored with one of the $k$ colors, the maximum-width RBCA problem
	can be solved in $O(n^3)$ time and $O(n^2)$ space.
\end{theorem}
%
%
We extend the above idea to compute a maximum-width $RBCA$ whose center lies on a given line $L$ in the plane.
We conclude this section with the following lemma.
\begin{lemma}\label{lem:cir_annu}
	A maximum-width RBCA \ann~whose center lies on a given query line $L$ in $\mathbb{R}^2$ can be computed in 
	$O(n^2)$ time and space.
\end{lemma}
\begin{proof}
	Using technique as described before to obtain Theorem~\ref{thm:un_cir_annu1}.
\end{proof}

\section{Conclusion}
In this paper, we have dealt with the problem of computing a maximum-width rainbow-bisecting empty annulus for three type of basic geometric objects among a set of given points on the plane. First, we propose an $O(n^3)$-time and $O(n)$-space algorithm for computing a maximum-width rainbow bisecting empty axis-parallel square annulus. Next, we compute a maximum width rainbow bisecting empty axis-parallel rectangular annulus in
$O(k^2n^2\log n)$-time and $O(n\log n)$-space. Finally, we propose a $O(n^3)$-time and $O(n^2)$-space algorithm for the maximum-width rainbow bisecting empty circular annulus problem. 

One of the interesting open problem in the square and rectangular setting is to study the arbitrary orientation case. Another obvious direction is to improve the running time. Further, the study of any non-trivial lower bound even for the basic problem that is computing the maximum width annulus (rectangle, square or circle) among a set of points on the plane would really be an interesting one.


{
\bibliographystyle{plain}
\bibliography{Ann}
}

\end{document}